\documentclass[11pt]{article}
\pdfoutput=1

\usepackage[T1]{fontenc}
\usepackage{microtype}
\usepackage{fullpage}
\usepackage{authblk}
\usepackage{amsthm}
\usepackage{amssymb}
\usepackage[english]{babel}
\usepackage[utf8]{inputenc}
\usepackage{amsmath}
\usepackage{amsfonts}
\usepackage{graphicx}
\usepackage{mathtools}
\usepackage{enumerate}
\usepackage{lineno}
\usepackage{thmtools, thm-restate}
\usepackage{cite}
\usepackage{todonotes}
\usepackage{bbm}
\usepackage{bm}
\usepackage{xspace}
\usepackage{url}
\usepackage{multicol}
\usepackage{booktabs}
\usepackage{tablefootnote}

\usepackage{color}
\definecolor{darkgreen}{rgb}{0,0.5,0}
\definecolor{darkblue}{rgb}{0,0,0.8}
\definecolor{darkred}{rgb}{0.8,0,0}

\usepackage{hyperref}
\hypersetup{
   unicode=false,          
   colorlinks=true,        
   linkcolor=darkred,          
   citecolor=darkblue,        
   filecolor=magenta,      
   urlcolor=brown           
}

\usepackage{array}

\usepackage{tikz}
\usetikzlibrary{shapes,shadows,arrows}

\newcommand\defeq{\stackrel{\mathclap{\normalfont{\mbox{\text{\tiny{def}}}}}}{=}}

\newtheorem{definition}{Definition}[section]
\newtheorem{lemma}[definition]{Lemma}
\newtheorem{theorem}[definition]{Theorem}

\newtheorem{corollary}[definition]{Corollary}

\graphicspath{{./figures/}}

\newcommand{\poly}{\mathrm{poly}}
\newcommand{\bigo}{\mathcal{O}}
\newcommand{\ignore}[1]{}

\newcommand{\indicator}[1]{\ensuremath{\mathbbm{1}\!\left[#1\right]}}
\newcommand{\floor}[1]{\ensuremath{{\lfloor}#1{\rfloor}}}
\newcommand{\ceil}[1]{\ensuremath{{\lceil}#1{\rceil}}}

\newcommand{\fodd}{\mathtt{odd}}
\newcommand{\feven}{\mathtt{even}}
\newcommand{\ham}{\mathsf{Ham}}
\newcommand{\thr}{\mathsf{Thr}}
\newcommand{\dom}{\mathsf{Dom}}
\newcommand{\newfootnote}[1]{\renewcommand*{\thefootnote}{\fnsymbol{footnote}}\footnote{#1}\renewcommand*{\thefootnote}{\arabic{footnote}}}

\newcommand{\vprod}{\textsc{vprod}}
\newcommand{\conv}{\textsc{conv}}
\newcommand{\mprod}{\textsc{mprod}}

\newcommand{\PMHamming}{\textsc{HammingDistancePatternMatching}\xspace}
\newcommand{\PMLessThan}{\textsc{LessThanPatternMatching}\xspace}
\newcommand{\PMThreshold}{\textsc{ThresholdPatternMatching}\xspace}

\newcommand{\PMLOne}{\textsc{{\ensuremath{\textsc{L}_{1}}}PatternMatching}\xspace}
\newcommand{\PMLTwo}{\textsc{{\ensuremath{\textsc{L}_{2}}}PatternMatching}\xspace}
\newcommand{\PMHammingShort}{\textsc{HamPM}\xspace}
\newcommand{\PMLessThanShort}{\textsc{LessThanPM}\xspace}
\newcommand{\PMThresholdShort}{\textsc{ThrPM}\xspace}
\newcommand{\PMLoddShort}{\textsc{\ensuremath{\textsc{L}_{2p+1}}\hspace{0cm}PM}\xspace}
\newcommand{\PMLevenShort}{\textsc{\ensuremath{\textsc{L}_{2p}}PM}\xspace}
\newcommand{\PMLOneShort}{\textsc{{\ensuremath{\textsc{L}_{1}}}PM}\xspace}
\newcommand{\PMLTwoShort}{\textsc{{\ensuremath{\textsc{L}_{2}}}PM}\xspace}

\newcommand{\AllPairsDominance}{\textsc{AllPairsDominanceProducts}\xspace}
\newcommand{\AllPairsLOneDistances}{\textsc{AllPairs\ensuremath{\textsc{L}_{1}}Distances}\xspace}
\newcommand{\AllPairsLTwoDistances}{\textsc{AllPairs\ensuremath{\textsc{L}_{2}}Distances}\xspace}

\newcommand{\AllPairsThreshold}{\textsc{AllPairsThresholdProducts}\xspace}
\newcommand{\AllPairsHamming}{\textsc{AllPairsHammingDistances}\xspace}
\newcommand{\AllPairsDominanceShort}{\textsc{APDom}\xspace}
\newcommand{\AllPairsLOneDistancesShort}{\textsc{AP\ensuremath{\textsc{L}_{1}}}\xspace}
\newcommand{\AllPairsLTwoDistancesShort}{\textsc{AP\ensuremath{\textsc{L}_{2}}}\xspace}
\newcommand{\AllPairsLoddDistancesShort}{\textsc{AP\ensuremath{\textsc{L}_{2p+1}}\hspace{0cm}}\xspace}
\newcommand{\AllPairsLevenDistancesShort}{\textsc{AP\ensuremath{\textsc{L}_{2p}}}\xspace}
\newcommand{\AllPairsThresholdShort}{\textsc{APThr}\xspace}
\newcommand{\AllPairsHammingShort}{\textsc{APHam}\xspace}

\newcommand{\sparse}{\textsf{Sparse}}

\title{\bf Hamming distance completeness\\ and sparse matrix multiplication}

\author{\Large Daniel Graf}
\author{Karim Labib}
\author{Przemys\l{}aw~Uzna\'nski}
\affil{Department of Computer Science, ETH Z\"urich, Switzerland}
\date{}
\begin{document}
\maketitle

\thispagestyle{empty}

\abstract{
We show that a broad class of $(+,\diamond)$ vector products (for binary integer functions $\diamond$) are equivalent under one-to-polylog reductions to the computation of the Hamming distance. Examples include: the dominance product, the threshold product and $\ell_{2p+1}$ distances for constant $p$.
Our results imply equivalence (up to polylog factors) between complexity of computation of All Pairs: Hamming Distances, $\ell_{2p+1}$ Distances, Dominance Products and Threshold Products.
As a consequence, Yuster's~(SODA'09) algorithm improves not only Matou\v{s}ek's (IPL'91), but also the results of Indyk, Lewenstein, Lipsky and Porat (ICALP'04) and Min, Kao and Zhu (COCOON'09).
Furthermore, our reductions apply to the pattern matching setting, showing equivalence (up to polylog factors) between pattern matching under Hamming Distance, $\ell_{2p+1}$ Distance, Dominance Product and Threshold Product, with current best upperbounds due to results of Abrahamson (SICOMP'87), Amir and Farach (Ann.~Math.~Artif.~Intell.'91), Atallah and Duket (IPL'11), Clifford, Clifford and Iliopoulous (CPM'05) and Amir, Lipsky, Porat and Umanski (CPM'05).
The resulting algorithms for $\ell_{2p+1}$ Pattern Matching and All Pairs $\ell_{2p+1}$, for $2p+1 = 3,5,7,\dots$ are new.

Additionally, we show that the complexity of \AllPairsHamming (and thus of other aforementioned \textsc{AllPairs}- problems) is within $\poly \log n$ from the time it takes to multiply matrices $n \times (n\cdot d)$ and $(n\cdot d) \times n$, each with $(n\cdot d)$ non-zero entries. This means that the current upperbounds by Yuster (SODA'09) cannot be improved without improving the sparse matrix multiplication algorithm by Yuster and Zwick~(ACM TALG'05) and vice versa.
}

\clearpage

\section{Introduction}

In the last few decades, many classical algorithmic problems received new attention when formulated as algebraic problems.
In pattern matching, instead of looking for occurrences of a pattern as a substring of a text, we can define a similarity score between two strings and ask for this score between the pattern $\mathbf{P}$ of length $m$ and every $m$-substring of the text $\mathbf{T}$ of length $n \geq m$. For example, scores of Hamming distance or $L_1$ distance 
 between numerical strings generalize the classical pattern matching: the total score for a given alignment is zero iff the pattern occurs exactly there in the text. One can go in this framework even further, and consider functions that are not metric, e.g. \PMLessThan which outputs the number of coordinates for which the pattern element is no larger than the corresponding text element. However, all those problems share a common additive structure, where for an input pattern $\mathbf{P}$ and text $\mathbf{T}$, the score vector $\mathbf{O}$ is such that $\mathbf{O}[i] = \sum_j \mathbf{P}[j] \diamond \mathbf{T}[i+j]$ for some binary function $\diamond$.
 
 Just as those pattern matching generalizations are based on \emph{convolution}, there is a family of problems based on \emph{matrix multiplication}, varying in flavour according to the vector product used. There we are given two matrices $A$ and $B$ representing $n$ vectors of dimension $d$, and the output is the matrix $O[i][j] = \sum_k A[i][k] \diamond B[k][j]$.
This is equivalent to the computation of all pairwise $(+,\diamond)$-vector products for two vector families, the so called \textsc{AllPairs-} problems.

In both of those worlds, the complexity is spanned between  \emph{easy} and \emph{hard} cases. An easy case is observed for e.g. $(+,\times)$ products, which have upper bound $\bigo(n \log n)$ for convolution (the classical Discrete Fast Fourier Transform algorithm) or $\bigo(n^\omega)$ (where  $\omega < 2.373$ c.f. Le~Gall~\cite{LeGall:2014:PTF:2608628.2608664}) for matrix multiplication. A hard case is considered to be respectively either near quadratic time or near cubic time problems. In the world of $(+,\diamond)$ vector products, we have not observed problems of the hard type, and instead they are either easy, or admit some intermediate complexity.\newfootnote{To observe good candidates for \emph{hard} problems, we have to go beyond $(+,\diamond)$ products, and consider  either $(\min,+)$ convolution (c.f. \cite{DBLP:conf/icalp/CyganMWW17,DBLP:conf/icalp/KunnemannPS17}) or $(\min,+)$ matrix product (c.f. \cite{DBLP:conf/focs/WilliamsW10}).} For many pattern matching generalizations there are independently achieved algorithms of identical complexity $\bigo(n \sqrt{m \log m})$. Similarly, for many AllPairs- problems, the best algorithms are of complexity $\bigo(n^{(\omega+3)/2})$ or similar. Why are so many different problems of essentially the same complexity?

\subsection*{Our contribution:}
We show that there is a shared source of hardness to those problems. That is, we show that for a wide class of $(+,\diamond)$ products, the corresponding problems are of (almost) equivalent hardness. This class includes not only products like Hamming distance or Dominance, but in fact any piecewise polynomial function of two variables (for appropriate definition of piecewise polynomiality, c.f. Definition~\ref{def:piecewise}) excluding certain degenerate forms (e.g. polynomials).  Thus we show that we should not expect the problems based on $(+,\diamond)$ products to be significantly harder to compute than e.g. ones based on Hamming distance (given reasonable restrictions on $\diamond$). The reduction applies both to Pattern Matching setting and to All Pairs- setting alike. We refer to Table~\ref{tab:problems} for a summary of considered problems, to Table~\ref{tab:prior-results} for a summary of existing upper bounds, and to Figure~\ref{fig:diag} for a summary of the old and new reductions.

Our reductions imply that any improvement made to one problem from the family translates to every other problem: e.g.~Yuster \cite{DBLP:conf/soda/Yuster09} improved the exponent of \AllPairsDominance (when $d=n$) to less than $(3+\omega)/2$ and this improvement applies to all other \textsc{AllPairs}- problems considered here. Another example is  that the  tradeoff achieved for one problem (e.g.\AllPairsHamming) between vectors dimension $d$ and the exponent (c.f. \cite{DBLP:conf/cocoon/MinKZ09} and \cite{DBLP:conf/isaac/GoldS17}) applies by our results to all the other \textsc{AllPairs}- problems considered here. Similarly, consider the sparsity of the input, where the tradeoff between the number of relevant entries in the input and the runtime (c.f. Vassilevska \cite{vassilevska2008efficient}, Vassilevska~et. al. \cite{DBLP:journals/toc/VassilevskaWY09} and Duan and Pettie \cite{DBLP:conf/soda/DuanP09}) applies to all of the mentioned problems. (See Section~\ref{sec:related} for precise upper bounds.)

We thus observe that there is a shared barrier in a broad class of problems and one is unlikely to improve upon existing upper bounds without some significant breakthrough. For both pattern matching problems and geometric problems we consider here, existing runtimes come from a tradeoff between the number of buckets and the size of these buckets. Without a novel technique, this runtime is unlikely to be improved. Similarly, any lower bound proof for one of the listed problems would immediately apply to every other problem. 

To exemplify this, we provide a conditional lowerbound 
to \AllPairsHamming (and thus to other \textsc{AllPairs}- problems) of the following form, linking its complexity to one of a sparse rectangular matrix multiplication (c.f. Theorem~\ref{th:sparsity_hamming}). First, we show that instance of \AllPairsHammingShort can be naturally expanded to an instance of matrix multiplication, with only 0/1 on input that is sparse. The reduction is straightforward, however it is interesting to observe that applying the fastest existing sparse matrix multiplication algorithm (c.f. Yuster and Zwick~\cite{DBLP:journals/talg/YusterZ05}) to resulting instance results in the same runtime as solving it directly (c.f. \cite{DBLP:conf/soda/Yuster09}). This hints that those problems have deeper connection, which we formalize by showing a converse reduction. Existence of such reduction is less obvious, since it has to handle more arbitrary structure of input matrices, restricted only by a total number of non-zero elements on the input. 

\renewcommand{\arraystretch}{1.2}

\begin{table}[t!]
\centering
{
\small
\begin{tabular}{llll}
\toprule
Name & Score function & Pattern Matching problem & All Pairs problem \\
\midrule
Hamming & $\indicator{x \not= y}$& $\mathbf{O}[i] = |\{j : \mathbf{P}[j] \not= \mathbf{T}[i+j] \}|$  &$O[i][j] = |\{ k : \mathbf{A}_i[k] \not= \mathbf{B}_j[k]\}|$\\
Dominance & $\indicator{x \le y}$& $\mathbf{O}[i] = |\{j : \mathbf{P}[j] \le \mathbf{T}[i+j] \}|$ & $O[i][j] = |\{ k : \mathbf{A}_i[k] \le \mathbf{B}_j[k]\}|$\\
Threshold & $\indicator{|x-y| \ge \delta}$ & $\mathbf{O}[i] = |\{j : |\mathbf{P}[j] - \mathbf{T}[i+j]| > \delta \}|$ & $O[i][j] = |\{k :  |\mathbf{A}_i[k] - \mathbf{B}_j[k]| > \delta\}|$\\
$\ell_1$ distance & $|x-y|$&$\mathbf{O}[i] = \sum_j |\mathbf{P}[j] - \mathbf{T}[i+j]|$ &$O[i][j] = \sum_{k=1}^n |\mathbf{A}_i[k] - \mathbf{B}_j[k] |$\\
$\ell_2$ distance &$(x-y)^2$ & $\mathbf{O}[i] = \sum_j (\mathbf{P}[j] - \mathbf{T}[i+j])^2$ &$O[i][j] = \sum_{k=1}^n (\mathbf{A}_i[k] - \mathbf{B}_j[k])^2$\\
\bottomrule
\end{tabular}
}
\caption{Summary of different score functions and the corresponding problems. $\indicator{\varphi}$ is $1$ iff $\varphi$ and $0$ otherwise.}
\label{tab:problems}
\vspace{-5mm}
\end{table}

\subsection*{Further applications:}
We provide the following further applications of our reductions. 
\begin{itemize}
\item Since they preserve structural properties of inputs like \emph{size of RLE}\newfootnote{Run Length Encoding}, in \cite{icalp2018GU} they were used to establish equivalence between runtime of \PMHamming and \PMLOne on instances with bounded RLE.
\item In Censor{-}Hillel et al. \cite{DBLP:journals/corr/abs-1802-04789} authors analyze complexity of sparse matrix multiplication under restricted bandwidth all-to-all communication model (the so called \textsc{Congested Clique} model). We note that their analysis implies immediately bounds to computation of All Pairwise Hamming Distance (and thus other vector products as well) by presented here Theorem~\ref{th:hamming_to_sparse} and  Theorem~\ref{th:sparse_to_hamming}.
\item Consider problem of Image Template Matching, where one is given as an input a two-dimensional text $T$ and a pattern $P$ of dimensions $n\times n$ and $m \times m$ respectively. The goal is to compute the dissimilarity score between $P$ and all $m \times m$-subsquares of $T$. Atallah in \cite{DBLP:journals/tip/Atallah01} gives the $\widetilde\bigo(n^2 m)$ runtime algorithm for $L_1$ version of this problem (so called \emph{Sum of the Absolute Value} difference measure). We note, that by our reductions, an equivalence between $L_1$, Hamming, Dominance and Threshold versions of this problem are established.
\end{itemize}

\begin{table}[t!]
\centering
{
\small
\begin{tabular}{lllllll}
\toprule
Name & \multicolumn{3}{c}{Pattern Matching problem} & 
\multicolumn{3}{c}{All Pairs problem}\\
\midrule
Hamming & \PMHammingShort & $\bigo(n \sqrt{m \log m})$ & \cite{Abrahamson87} & \AllPairsHammingShort & $\bigo(n^{(\omega+3)/2})$& \cite{DBLP:conf/cocoon/MinKZ09}\\
Dominance & \PMLessThanShort & $\bigo(n \sqrt{m \log m})$ & \cite{Amir1991} & \AllPairsDominanceShort & $\bigo(n^\rho)$&\cite{DBLP:conf/soda/Yuster09}\\
Threshold & \PMThresholdShort & $\bigo(n \sqrt{m \log m})$ & \cite{DBLP:journals/ipl/AtallahD11} & \AllPairsThresholdShort & $\bigo(n^{(\omega+3)/2})$& \cite{DBLP:conf/icalp/IndykLLP04} \\
$\ell_1$ distance & \PMLOneShort & $\bigo(n \sqrt{m \log m})$ & \cite{DBLP:conf/cpm/CliffordCI05,Amir2005} & \AllPairsLOneDistancesShort & $\bigo(n^{(\omega+3)/2})$ & \cite{DBLP:conf/icalp/IndykLLP04}\\
$\ell_2$ distance & \PMLTwoShort & $\bigo(n \log m)$ & \cite{DBLP:journals/algorithmica/LipskyP11} & \AllPairsLTwoDistancesShort & $\bigo(n^\omega)$ &  \cite{DBLP:conf/icalp/IndykLLP04} \\
\bottomrule
\end{tabular}
}
\caption{Overview of the known results and of how we abbreviate the corresponding problem names. $\rho \le 2.6834$ is a solution to $\rho = \omega(1, 4-\rho, 1)$, where $\omega(a,b,c)$ is the exponent of fast multiplication of rectangular matrices $n^a \times n^b$ with $n^b \times n^c$.}
\label{tab:prior-results}
\vspace{-5mm}
\end{table}

\section{Related work}
\label{sec:related}
We now list different pattern matching problems that differ in their underlying score functions. 
The \textbf{\PMHamming}\newfootnote{Also known as Pattern Matching with Mismatches.} was studied by Abrahamson~\cite{Abrahamson87}, \textbf{\PMLessThan}  was introduced by Amir and Farach~\cite{Amir1991}, \textbf{\PMLOne} was studied independently by Clifford, Clifford and Iliopoulous~\cite{DBLP:conf/cpm/CliffordCI05} and Amir, Lipsky, Porat and Umanski~\cite{Amir2005} (although 2-dimensional version of this problem was studied earlier by Atallah~\cite{DBLP:journals/tip/Atallah01}) and \textbf{\PMThreshold} was studied by Atallah and Duket~\cite{DBLP:journals/ipl/AtallahD11}. 
For all those problems, the currently best known algorithms run in time $\bigo(n \sqrt{m \log m})$ using similar techniques: high/low frequency, bucketing and convolution. 

For \textbf{\PMLTwo}, since $\mathbf{O}[i] = \sum_j \mathbf{P}[j]^2 + \sum_j  \mathbf{T}[i+j]^2 - 2 \sum_j \mathbf{P}[j] \mathbf{T}[i+j]$, 
 the dominating term in the computation arises from computing a single convolution in time $\bigo(n \log m)$ via Fast Fourier Transform (FFT) as observed by Lipsky and Porat~\cite{DBLP:journals/algorithmica/LipskyP11}. This approach extends to \PMLevenShort, with runtime $\bigo(p^2 n \log m)$.

On the side of reductions, only little was known. Lipsky and Porat~\cite{DBLP:journals/ipl/LipskyP08a} showed that both \PMHammingShort and \PMLessThanShort reduce to \PMLOneShort showing that the latter problem is no easier than the former problems. The question of whether e.g.~\PMHammingShort could be substantially easier than \PMLOneShort remained open. The first non-trivial reduction (although not stated as a lower-bound type result) was provided by Zhang and Atallah~\cite{DBLP:journals/ipl/ZhangA17}, where they showed that \PMThresholdShort with threshold $\delta$ reduces to $\bigo(\log \delta)$ instances of \PMHammingShort(see Figure \ref{fig:diag}).

In computational geometry, a classical problem is to process a set of $n$ points given in $d$-dimensional space. One can consider e.g.~the metric space and ask for a pair of closest or farthest points. Some of those problems in low-dimensions (i.e.~$d = \poly \log n$) exhibit a structure that allows for solutions almost linear in $n$ for some metrics (see Williams~\cite{DBLP:conf/soda/Williams18}). However, in high-dimensional data, the situation is usually dire, as the so called \emph{curse of dimensionality} kicks in (c.f. \cite{DBLP:reference/ml/KeoghM17} and \cite{DBLP:journals/toc/Har-PeledIM12}) and for processing such spaces usually the fastest known approach is to compute all pairwise distances \cite{DBLP:conf/icalp/IndykLLP04}.

\textbf{\AllPairsDominance} was introduced by Matou\v{s}ek~\cite{dominance}, where he provided a solution working in time $\bigo(n^{(\omega+3)/2}) \subseteq \bigo(n^{2.687})$. 
Vassilevska~\cite{vassilevska2008efficient} and Vassilevska et al.~\cite{DBLP:journals/toc/VassilevskaWY09} considered dominance product on sparse inputs where we denote by $m_1$ and $m_2$ the number of entries in $A$ and $B$, respectively that contribute to the score. 
They obtain a bound of $\bigo(\min(n^\omega+\sqrt{m_1 m_2} \cdot n^{\frac{\omega - 1}{2}},n^2 + (m_1 m_2)^{\frac{\omega - 2}{\omega - \alpha - 1}} n^{\frac{2 - \alpha \omega}{\omega - \alpha - 1}})).$\newfootnote{$\alpha = \sup\{0 \le r \le 1 : \omega(1,r,1) = 2+o(1)\} \ge 0.31389$.}
Duan and Pettie~\cite{DBLP:conf/soda/DuanP09} simplified this analysis.
For $d \ll n$, Vassilevska and Williams \cite{Vassilevska:2006:FMW:1132516.1132550} and \cite{vassilevska2008efficient} gave an algorithm with a time bound of $\bigo(n^{\frac{2\omega - \omega \alpha - 2}{\omega - \alpha -1}} d^{\frac{2\omega - 4}{\omega - \alpha - 1}} + n^{2 + o(1)})$. 
Yuster~\cite{DBLP:conf/soda/Yuster09} improved the bound of the case $d = n$ to $\bigo(n^\rho)$, where $\rho$ is a solution to $\rho = \omega(1, 4-\rho, 1)$. The bound $\rho \le 2.6834$ is provided.
Recently, Gold and Sharir~\cite{DBLP:conf/isaac/GoldS17} presented an updated analysis of the time vs.~dimension tradeoff using newer bounds on rectangular matrix multiplication. For $d=n$, this gives a runtime of $\bigo(n^{2.6598})$.

\textbf{\AllPairsLOneDistances} was considered by Indyk et al. \cite{DBLP:conf/icalp/IndykLLP04},with an $\bigo(n^{(\omega+3)/2})$ algorithm for the case when $d=n$.
Although not stated as such, one algorithm presented in \cite{DBLP:conf/icalp/IndykLLP04} can be adapted to computing \textbf{\AllPairsThreshold}  in time $\bigo(n^{(3+\omega)/2})$. \cite{vassilevska2008efficient} introduced  \textbf{$\bm{(+,\max)}$-\textsc{MatrixProduct}} where the score matrix is $O[i][j] = \sum_{k=1}^n \max(\mathbf{A}_i[k],\mathbf{B}_j[k])$ and presented a bucketing solution with runtime $\bigo(n^{(\omega + 3) / 2})$. The algorithm follows in spirit the \AllPairsLOneDistancesShort algorithm from \cite{DBLP:conf/icalp/IndykLLP04} with a tweaked score contribution.
\textbf{\AllPairsHamming} was considered by Min et al.~\cite{DBLP:conf/cocoon/MinKZ09}. Inspired by the reduction from Hamming distance to $L_1$ in \cite{DBLP:journals/ipl/LipskyP08a}, they utilized the \AllPairsLOneDistancesShort algorithm from \cite{DBLP:conf/icalp/IndykLLP04}. This resulted in a $\bigo(n^{(\omega+3)/2})$ time algorithm when $d=n$. They also utilized rectangular matrix multiplication bounds to provide a tradeoff in the complexity when $d \ll n$. Writing their upper bound in a general form, the complexity is $\bigo(n^{1 + \omega(1,s,1) / 2}\sqrt{d})$ where $d = n^s$. Given the  improved bounds for rectangular matrix multiplication by Le~Gall~\cite{DBLP:conf/focs/Gall12} and subsequently by Le~Gall and Urrutia~\cite{DBLP:conf/soda/GallU18}, the bounds from \cite{DBLP:conf/cocoon/MinKZ09} are stronger.
\textbf{\AllPairsLTwoDistances} as observed by Indyk et al. \cite{DBLP:conf/icalp/IndykLLP04} reduces to a single matrix multiplication and thus admits a runtime of $\bigo(n^\omega)$. Similarly, \AllPairsLevenDistancesShort admits a runtime of $\bigo(p^2 n^\omega)$.

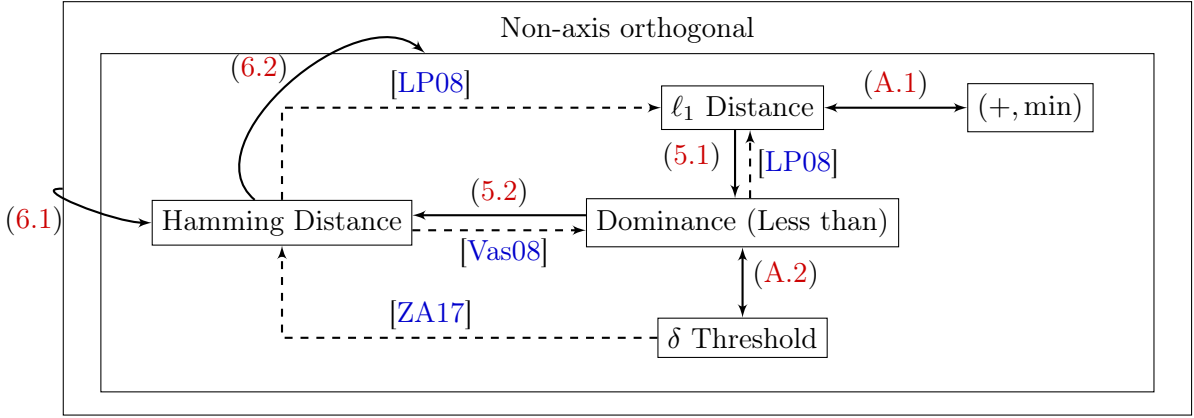
\begin{figure}[t!]
\tikzstyle{line} = [draw, -latex', thick, dashed]
\tikzstyle{lineNew} = [draw, -latex', thick]
\tikzstyle{block} = [draw, rectangle, text centered, node distance = 4em]
\tikzstyle{bigRectangle} = [draw, rectangle, minimum width = 15cm, minimum height = 5.5cm, font=\Large]

\tikzstyle{smallRectangle} = [draw, rectangle, minimum width = 14 cm, minimum height = 4.5cm, font=\Large]

\begin{tikzpicture}
\node [bigRectangle, anchor=center, label=Piecewise Polynomials] (PieceWise){};

\node[smallRectangle, anchor=center, label=Non-axis orthogonal, yshift= -0.5em] at (PieceWise.center) (NonAxisOrth){};
\node [block, xshift=4em] at (NonAxisOrth.center)(Dominance) {Dominance (Less than)};
\node [block, above of=Dominance] (L1Distance) {$\ell_1$ Distance};
\node [block, below of=Dominance] (Threshold) {$\delta$ Threshold};

\node [block, left of=Dominance,xshift=-12em] (HammingDistance) {Hamming Distance};
\node [block, right of=L1Distance,xshift=6em] (PlusMin) {$(+,\min)$};

\path [line] (HammingDistance) |- node[xshift=5 em,yshift=3mm]{\cite{DBLP:journals/ipl/LipskyP08a}}(L1Distance);
\path [line] ([yshift=-1mm]HammingDistance.east) -- node[xshift=0 em,yshift=-3mm]{\cite{vassilevska2008efficient}} ([yshift=-1mm]Dominance.west);
\path [line] ([xshift=1mm]Dominance.north) -- node[xshift=1.5 em]{\cite{DBLP:journals/ipl/LipskyP08a}} ([xshift=1mm]L1Distance.south);
\path [line] (Threshold) -| node[xshift = 5em, yshift=3mm]{\cite{DBLP:journals/ipl/ZhangA17}} (HammingDistance);

\path [lineNew] ([xshift=-1mm]L1Distance.south) -- node[xshift=-1.5 em,yshift=1mm]{\eqref{th:L1_to_Dominance}} ([xshift=-1mm]Dominance.north);	
\path[lineNew] (PieceWise) to[out=-182,in=180] node[xshift=-6mm,yshift=-2mm]{\eqref{th:completeness1}} (HammingDistance);
\path[lineNew] (HammingDistance) to[out=140,in=140] node[xshift=-4mm,yshift=2mm]{\eqref{th:completeness2}} (NonAxisOrth);

\path [draw, latex'-latex', thick] (L1Distance.east) -- node[xshift=0 em,yshift=3mm]{\eqref{lem:l1_equiv_min}} (PlusMin.west);	
\path [lineNew] ([yshift=1mm]Dominance.west) -- node[xshift=0 em,yshift=3mm]{\eqref{th:Dominance_to_Hamming}} ([yshift=1mm]HammingDistance.east);
\path [draw, latex'-latex', thick] (Threshold.north) -- node[xshift=1.5 em,yshift=1mm]{\eqref{lem:Dom_eq_Thr}} (Dominance.south);
\end{tikzpicture}
\caption{Existing and new reductions between problems, together with problem classes.}
\label{fig:diag}
\vspace{-4mm}
\end{figure}

We observe that $L_2$ is an ''easy'' score function. For every other score function mentioned, all solutions presented use a bucketing or a high/low frequency technique to decompose the problem into ones solvable by convolution (for Pattern Matching problems) or matrix multiplication (for All Pairs problems).
We refer to Tables~\ref{tab:problems} and~\ref{tab:prior-results} for a summary.
There are several related problems that use the aforementioned problems as subroutines. 
\textbf{Weighted Pattern Matching} in the most general setting asks for $\mathbf{O}[i] = \sum_j w(P[j],T[i+j])$ for some weight function $w$. In \cite{DBLP:journals/algorithmica/LipskyP11} authors presented a simple $\bigo(|\Sigma| n \log m)$ algorithm.
\textbf{Pattern Matching with Wildcards} admits a simple determinsitic $\bigo(n \log m)$ solution via weighted \PMLTwoShort, as shown by Clifford and Clifford~\cite{DBLP:journals/ipl/CliffordC07}.
\textbf{Closest $L_\infty$ Pair} was considered by Indyk et al. \cite{DBLP:conf/icalp/IndykLLP04} where the presented algorithm uses binary search on top of \AllPairsThresholdShort. The total runtime is $\bigo(n^{(\omega+3)/2} \log D)$, where $D$ is the diameter of the input point set.

It is possible to speed up matrix multiplication beyond  $\bigo(n^\omega)$ time when the input matrices are sparse, i.e. out of $n^2$ entries, they have only $m_1$ and $m_2$ entries that are non-zero. If $m_1$ and $m_2$ are sufficiently small, a faster algorithm is possible. Yuster and Zwick~\cite{DBLP:journals/talg/YusterZ05} presented an algorithm with runtime
$\bigo(\min((m_1 m_2)^\frac{ \omega-2}{ \omega-1-\alpha} n^{\frac{2 - \alpha\omega}{\omega-1-\alpha}} + n^{2+o(1)}, m_1 n, m_2 n, n^{\omega}))$. 
We denote $\sparse(a,b,c;m_1,m_2)$ as the time of optimal algorithm for multiplication sparse matrices $a\times b$ and $b\times c$, with $m_1$ and $m_2$ nonzero entries respectively, so for instance \cite{DBLP:journals/talg/YusterZ05} states upperbound on $\sparse(n,n,n;m_1,m_2)$.

\section{Preliminaries}
Previously discussed problems have at their core the computation of $(+,\diamond)$ vector product, that is $\sum_i x_i \diamond y_i$ for some binary function $\diamond$. 
Formally, for vectors $\mathbf{A},\mathbf{B}$ and matrices $\mathcal{A},\mathcal{B}$, we denote
the $(+,\diamond)$ vector product as $\vprod(\diamond, \mathbf{A}, \mathbf{B}) \ \defeq\  \sum_{i} \mathbf{A}[i] \diamond \mathbf{B}[i]$, the $(+,\diamond)$ convolution as $\conv(\diamond, \mathbf{A},\mathbf{B}) = \mathbf{C} \text{\ \ such that\ \ } \mathbf{C}[k] = \sum_{i+j = k} \mathbf{A}[i] \diamond \mathbf{B}[j]$ and
the $(+,\diamond)$ matrix product as $\mprod(\diamond,\mathcal{A},\mathcal{B}) = \mathcal{C} \text{\ \ such that\ \ } \mathcal{C}[i,j] = \sum_{k} \mathcal{A}[i,k] \diamond \mathcal{B}[k,j]$.

Thus, e.g.~since $\ham(x,y) \defeq \indicator{ x \not= y }$,
then  $\vprod(\ham,\mathbf{X},\mathbf{Y})$ is the Hamming Distance between $\mathbf{X}$ and $\mathbf{Y}$, \PMHammingShort is essentially $\conv(\ham,\mathbf{X},\mathbf{Y}^R)$, and \AllPairsHammingShort between vectors $\{X_1,\ldots,X_n\}$ and $\{Y_1,\ldots,Y_n\}$ is $\mprod\Big(\ham,\begin{bmatrix} X_1& \cdots& X_n \end{bmatrix}^T, \begin{bmatrix} Y_1& \cdots& Y_n \end{bmatrix}\Big)$. 

We now shift our attention to the relations between the binary functions.
\begin{definition}
\label{def:reduction}
We say that $\diamond$ reduces preserving linearity to instances of $\square_1, \dots,\square_K$, if there are functions $f_1,\ldots,f_K$ and $g_1,\ldots,g_K$ and coefficients $\alpha_1,\ldots,\alpha_K$, such that for any $x,y$:\newfootnote{For the sake of simplicity, we are omitting in the definition the post-processing function necessary e.g.~$(\ \cdot\ )^{1/p}$ for $L_p$ norms.}
$$
x \diamond y = \sum_{i} \alpha_i \cdot \Big(f_i(x)\ \square_i\ g_i(y)\Big).
$$
\end{definition}
A one-to-many reduction from $\diamond$ to $\square$ is also a one-to-many reduction from $(+,\diamond)$ vector product/convolution/matrix multiplication to $(+,\square)$ vector product/convolution/matrix multiplication. Indeed, given Definition \ref{def:reduction}, we have for any vectors $\mathbf{A},\mathbf{B}$ and matrices $\mathcal{A},\mathcal{B}$:
$\vprod(\diamond,\mathbf{A},\mathbf{B}) = \sum_{i} \alpha_i \cdot \vprod(\square_i,f_i(\mathbf{A}),g_i(\mathbf{B}))$,
$\conv(\diamond,\mathbf{A},\mathbf{B}) = \sum_{i} \alpha_i \cdot \conv(\square_i,f_i(\mathbf{A}),g_i(\mathbf{B}) )$ and
$\mprod(\diamond,\mathcal{A},\mathcal{B}) = \sum_{i} \alpha_i \cdot \mprod(\square_i,f_i(\mathcal{A}), g_i(\mathcal{B}))$, where $f(\mathbf{A})$ and $f(\mathcal{A})$ denotes a coordinate-wise application of $f$ to vector $\mathbf{A}$ and matrix $\mathcal{A}$, respectively.

\section{Main results}

\emph{Remark:}
We assume that all input values and coefficients are integers bounded in absolute value by $\poly(n)$.

\begin{definition}
\label{def:piecewise}
For integers $A,B,C$ and polynomial $P(x,y)$ we say that the function $P(x,y) \cdot \indicator{A x + B y + C > 0}$ is \emph{halfplane polynomial}.
We call a sum of halfplane polynomial functions a \emph{piecewise polynomial}.
We say that a function is \emph{axis-orthogonal piecewise polynomial}, if it is piecewise polynomial and for every $i$, $A_i = 0$ or $B_i = 0$.
\end{definition}

Observe that $\ham(x,y) = \indicator{x > y} + \indicator{x < y}$, $\max(x,y) = x \cdot \indicator{x \ge y} + y \cdot \indicator{x < y}$, $|x-y|^{2p+1} = (x-y)^{2p+1} \cdot \indicator{x > y} + (y-x)^{2p+1} \cdot \indicator{x < y}$, and $\thr_\delta(x,y) \defeq \indicator{ \lvert x - y \rvert \geq \delta} = \indicator{x \leq y-\delta} + \indicator{x \geq y+\delta}$.

\begin{theorem}
\label{th:intro2}
Let $\diamond$ be a piecewise polynomial of constant degree and $\poly \log {n}$ number of summands.
\begin{itemize}
\item If $\diamond$ is axis orthogonal, then $\diamond$ is ``easy'': $(+,\diamond)$ convolution takes $\widetilde{O}(n)$ time, $(+,\diamond)$ matrix multiplication takes $\widetilde{O}(n^\omega)$ time.
\item Otherwise, $\diamond$ is \emph{Hamming distance complete}: under one-to-polylog reductions, on inputs bounded in absolute value by $\poly(n)$, $(+,\diamond)$ product is equivalent to Hamming distance, $(+,\diamond)$ convolution is equivalent to \PMHammingShort and $(+,\diamond)$ matrix multiplication is equivalent to \AllPairsHammingShort.
\end{itemize}
\end{theorem}
Theorem~\ref{th:intro2} follows from two technical results presented in Section~\ref{sec:hdc}, Theorem~\ref{th:completeness1} and Theorem~\ref{th:completeness2}.

\begin{corollary}
\label{cor:intro1}
The following problems are equivalent under one-to-polylog reductions: \PMHammingShort, \PMLessThanShort, \PMLoddShort for a constant integer $p\ge0$, \PMThresholdShort and $(+,\max)$-\textsc{Convolution}.
\end{corollary}
\begin{corollary}
\label{cor:intro2}
The following problems are equivalent under one-to-polylog reductions: \AllPairsHammingShort, \AllPairsDominanceShort, \AllPairsLoddDistancesShort for a constant integer $p\ge0$, \AllPairsThresholdShort and $(+,\max)$-\textsc{MatrixProduct}.
\end{corollary}


Taking advantage of the completeness of \AllPairsHammingShort among all considered \textsc{AllPairs}- problems, we link the complexity of all those problems to the problem of multiplying of rectangular sparse matrices.
\begin{theorem}
\label{th:sparsity_hamming}
The time complexity of \AllPairsHammingShort on $n$ vectors of dimension $d$ is (under randomized Las Vegas reductions)
within $\poly \log n$ from $\sparse(n,\min(d^2,nd),n;nd;nd)$.
\end{theorem}

\section{Warm-up}

We start by showing a reduction from $L_1$ distance to $\bigo(\log^2 n)$ instances of Hamming distance. This is a reduction  that is fully contained in our general reduction, that is Theorem~\ref{th:intro2}. However, since it uses similar techniques, and already has a nontrivial consequence (e.g.~collapsing hardness of \PMLOneShort and \PMHammingShort), we present it separately.
\paragraph*{Scaling:} Observe that for many ``natural'' functions $\diamond$ and integers $x,y$, $x \diamond y$ is approximated by $\lfloor x/2 \rfloor \diamond  \lfloor y/2 \rfloor$ (up to some fixed multiplicative factor). This allows us to unwind $x \diamond y$ into a weighted sum of $\bigo(\log (\max(|x|,|y|))$ corrective terms. For example, if for some constant $C$, integers $x,y \ge 0$ and some corrective function $\xi$:\quad
$x \diamond y = C \cdot (\lfloor x/2 \rfloor \diamond  \lfloor y/2 \rfloor) + \xi(x,y)$
then naturally
$
x \diamond y = 0 \diamond 0 + \sum_{i \ge 0} C^i \cdot \xi(\lfloor x/2^i \rfloor, \lfloor y/2^i \rfloor).
$
\paragraph*{Sparsity:}
We consider a generalized version of the input with special ``ignore'' marks $\star$ as possible elements. Those elements of the input never contribute to the final score of the $(+,\diamond)$ product. Formally, we operate on $\mathbb{Z}+\{\star\}$, with special arithmetic rules (unless stated otherwise):
\begin{itemize}
\item for any \emph{single argument} function: $f(\star) = \star$,
\item for any \emph{double argument} function: $g(\star,\star) = g(\star,y) = g(x,\star) = 0$.\newfootnote{We have to keep in mind that whether a function is a single or double argument is context dependent: e.g. writing: $\indicator{x\not=y}\  =\  1 - \indicator{x=y}$, we have to treat $1$ as a function of $x$ and $y$ as well.}
\end{itemize}
The goal of this formalism is twofold. The first one is to handle sparse inputs formally (i.e.~vectors with $\bigo(n^{1-\varepsilon})$ relevant entries). The second one is that such ``ignore'' marks coupled with filtering (defined below) allows us to split the input based on properties of its values.
We note that these ``ignore'' marks do not increase the computational complexity of Hamming distance (see Lemma~\ref{lem:nonsparse_hamming} in the Appendix).
\paragraph*{Filtering:}
We define the following functions: 
$$\feven(x) \defeq \begin{cases}x \quad\text{ if }x\text{ is even}\\ \star\quad\text{ otherwise } \end{cases} \quad\quad \fodd(x) \defeq \begin{cases}x \quad\text{ if }x\text{ is odd}\\ \star\quad\text{ otherwise } \end{cases}$$

Those functions, when applied to a vector or a matrix, allows us to filter values according to parity, e.g. for $\mathbf{A} = [1,2,3,4]$ one gets $\feven(\mathbf{A}) = [\star,2,\star,4]$.

We now give two reductions that illustrate the usefulness of these techniques.
Both reductions are illustrated in the Appendix~\ref{sec:examples} (Figures~\ref{fig:example-l1-to-less-than} and~\ref{fig:example-less-than-to-ham}).
\begin{theorem}\label{th:L1_to_Dominance}
The $L_1$ distance reduces to $\bigo(\log n)$ instances of dominance.
\end{theorem}
\begin{proof}

Since $L_1$ distance is shift-invariant, i.e. $|(x+\Delta) - (y+\Delta)| = |x-y|$ for any $\Delta$, we can assume that $0 \le x,y < M$ for some $M = \poly(n)$.
Observe that for $x,y \ge 0$, $
	|x - y| = 2 \cdot \Big|\floor{x/2} - \floor{y/2}\Big| + \eta(x,y),
	$ 	where, denoting $\dom(x,y) \defeq \indicator{x \le y}$,
\begin{align*}
\eta(x,y) &= \indicator{(x\text{ is odd}) \wedge(y\text{ is even}) \wedge (x\geq y)} - \indicator{(x\text{ is even}) \wedge(y\text{ is odd}) \wedge (x\geq y)}\\ 
&+\indicator{(y\text{ is odd}) \wedge(x\text{ is even}) \wedge (y\geq x)} - \indicator{(y\text{ is even}) \wedge(x\text{ is odd}) \wedge (y\geq x)}\\
&=\dom(\fodd(-x), \feven(-y)) - \dom(\feven(-x),\fodd(-y) )\\
&+\dom( \feven(x), \fodd(y)) - \dom( \fodd(x), \feven(y)).
\end{align*}
By unwinding, we get $|x - y| = \sum_{i=0}^{{\log M} } 2^i \cdot \eta( \floor{x/2^i}, \floor{y/2^i} )$ which completes the reduction.
\end{proof}

\begin{theorem}
	\label{th:Dominance_to_Hamming}
Dominance reduces to $\bigo(\log n)$ instances of Hamming distance and multiplication.
\end{theorem}
\begin{proof}
Since dominance is shift-invariant, w.l.o.g. we assume that $0 \le x,y < M$ for some $M = \poly(n)$.
Observe the following recurrence relation, for $x,y \ge 0$:
\begin{align*}
\dom(x,y) &= \dom(\floor{x/2}, \floor{y/2}) - \indicator{(x\text{ is odd})  \wedge (x = y+1)}\\ 
&= \dom(\floor{x/2}, \floor{y/2}) - \indicator{x\text{ is odd}} + \indicator{x\text{ is odd}}\cdot \ham(x,y+1)
\end{align*}
By unwinding, we get:
$$
\dom(x,y) = 1 - \sum_{i=0}^{{\log M} } \indicator{\lfloor x/2^i \rfloor \text{ is odd}} + \sum_{i=0}^{{\log M}} \indicator{\lfloor x/2^i \rfloor \text{ is odd}} \cdot \ham(\lfloor x/2^i \rfloor, \lfloor y/2^i \rfloor+1).
$$
Using filtering notation, this becomes
$$
\dom(x,y)\quad =\quad\underbrace{1 - \sum_{i=0}^{{\log M}} \indicator{\lfloor x/2^i \rfloor \text{ is odd}}}_{(\ast)} +\underbrace{\sum_{i=0}^{{\log M}} \ham(\fodd(\lfloor x/2^i \rfloor), \lfloor y/2^i \rfloor+1)}_{(\ast\ast)}
$$

Now observe, that $(\ast)$ is purely a function of $x$. If $x$ is guaranteed to be an integer, then evaluating it as part of an operator (i.e. inside convolution or matrix-multiplication) is trivial.
As $y$ is never mapped to $\star$ in $(\ast\ast)$, treating $(\ast)$ as a single argument function suffices.

The second term $(\ast\ast)$ uses our filtering function and the convention that $\ham$ evaluates to $0$ if at least one of its inputs is $\star$. Thus $(\ast\ast)$ is a sum of $\bigo(\log n)$ Hamming distances on inputs from $\mathbb{Z}\cup\{\star\}$. By Lemma~\ref{lem:nonsparse_hamming}, each of those reduces to two instances of Hamming distance on inputs from $\mathbb{Z}$.
\end{proof}

\emph{Remark:}
In general, we have to take into account that both $x,y \in \mathbb{Z}\cup \{\star\}$.
Thus, we have to treat term $(\ast)$ as a function of \emph{both} $x$ and $y$, that is evaluating to $0$ if $x = \star$ or $y = \star$.
In general, $(\ast)$ reduces to evaluating, after the reduction step, some polynomial $Q(x',y') = f(x')$ (where $y'$ might be $\star$) with $f(x') = 1 - \sum_{i=0}^{\log M} \indicator{\lfloor x'/2^i \rfloor \text{ is odd}}$.
By Lemma~\ref{th:polynomials_are_easy}, $f(x')$ can be done in the time of a regular convolution or matrix multiplication
and thus the computation time for $(\ast)$ is dominated by $(\ast\ast)$, that is \PMHammingShort and \AllPairsHammingShort, respectively.

\section{Hamming distance completeness}
\label{sec:hdc}
The goal of this section is to prove Theorem~\ref{th:intro2}. We achieve this by showing two separate reductions, one from \emph{all piecewise polynomial functions} to Hamming distance and one from Hamming distance to all non axis-orthogonal piecewise polynomials.

\begin{theorem}
\label{th:completeness1}
If $\diamond$ is a piecewise polynomial of degree $d$ with $c$ summands then it reduces to $\bigo(c \cdot d^2 \cdot \log^{d+1} n)$ instances of Hamming distance. Reduction works even if we allow ``don't care'' symbols.
\end{theorem}

\begin{theorem}
\label{th:completeness2}
If $\diamond$ is a piecewise polynomial of degree $d$ but is not axis-orthogonal piecewise polynomial, then Hamming distance reduces to $\bigo(d^2)$ instances of $\diamond$ and multiplication.
\end{theorem}

To prove Theorem~\ref{th:completeness1}, we  consider every summand separately. We show that summands with ``simple'' conditions (that depend on only one argument) are no harder than simple multiplication. Every other summand with conditional term $\indicator{A_i x + B_i y + C_i > 0}$ reduces under linear transformations of its arguments to $\indicator{x < y}$. It is thus enough to consider terms of the form $x^a y^b \indicator{x < y}$. We decompose such terms recursively into a sum of: terms with smaller values ($x/2, y/2$ instead of $x,y$), terms of smaller degree, and terms with a conditional term of a simpler form of $\indicator{x = y}$. Exhaustively applying this decomposition  leaves us with a polylog number of terms of the form $w(x) \cdot  \indicator{x=y}$, with which we deal separately (those decompose into a logarithmic number of regular Hamming distances).

\begin{lemma}
\label{lem:weighted_hamming}
	For an integer weight function $w$, the character weighted matches, that is $w(x)\cdot \indicator{x=y}$, reduce to $\bigo(\log n)$ instances of Hamming distance and multiplication.
\end{lemma}

\begin{proof}
Let $W$ be the upper bound on all values of $w$ in the considered domain of inputs.
	Given two integers $x,y$, we observe the following equality:
	$$
	w(x) \cdot [x = y]\quad=\quad \sum_{i=0}^{{\log W}} 2^i \cdot \indicator{w_i(x) = w_i(y)}
	$$
	where the filtering function $w_i$ is defined based on $w$:
 $$w_i(x) = \begin{cases}x&\quad i\text{-th bit of }w(x)\text{ is }1\\\star&\quad\text{otherwise.}\end{cases} $$
 Observing that $\indicator{x=y} = 1 - \ham(x,y)$ finishes the proof.
\end{proof}

\begin{lemma}
\label{th:polynomials_are_easy}
An axis-orthogonal piecewise polynomial $\diamond$ of $c$ summands of degree $d$ reduces to $\bigo(d^2 c)$ multiplications.
\end{lemma}

\begin{proof}
	Given an axis orthogonal piecewise polynomial $F(x,y) =\sum_{i = 1}^{c}P_i(x,y)\cdot \indicator{A_ix + B_iy + C_i > 0}$ of degree $d$. Consider summand $P_i(x,y) \indicator{A_i x + C_i > 0}$ (w.l.o.g. we assume that $B_i = 0$ and $A_i \not=0$).  Consider a monomial of $P_i(x,y)$, e.g. $x^ay^b$. Define $x' = x^a$ iff $A_i x + C_i > 0$ and $x' = \star$ otherwise, and $y' = y^b$. Then $x^ay^b \cdot \indicator{A_i x + C_i > 0} = x' y'$.
\end{proof}
Axis orthogonal piecewise polynomial $\diamond$ are no harder than multiplication  in e.g. vector convolution or matrix multiplication. By Lemma~\ref{th:polynomials_are_easy} it reduces to multiplication in $\mathbb{Z} \cup \{\star\}$, which in turn reduces to multiplication in $\mathbb{Z}$. Indeed, it is enough to consider a map $\mathbb{Z} \cup \{\star\} \to \mathbb{Z}$ that is identity on $\mathbb{Z}$ and maps $\star \to 0$.

\begin{lemma}
\label{th:monomials}
Given integers $a,b \ge 0$, the binary function $x^a y^b \cdot\indicator{x<y}$ reduces to $\bigo(\log^{a+b+1} n)$ instances of Hamming distance and multiplication.
\end{lemma}

\begin{proof}
Denote $\text{MDom}_{a,b}(x,y) = x^a y^b \cdot\indicator{x<y}$, $\text{MEq}_{a}(x,y) = x^a \cdot \indicator{x=y}$.
First, we argue that w.l.o.g. $x,y \ge 0$. Indeed, observe that $\text{MDom}_{a,b}(x - \Delta,y - \Delta) = (x-\Delta)^a (y-\Delta)^b \cdot\indicator{x<y}$, thus for large enough $\Delta$, the computation of $\text{MDom}_{a,b}$ on inputs of arbitrary sign reduces to $\bigo(ab)$ instances of $\text{MDom}$ on non-negative inputs. Thus we assume that $0 \le x,y \le M$ for some $M = \poly(n)$.

We proceed with the following decomposition, where $u = \lfloor \frac{x}2 \rfloor$ and $v = \lfloor \frac{y}2 \rfloor$. 
\begin{align}
\label{eq:term1}\tag{*}\text{MDom}_{a,b}(x,y)\quad =\quad &(2u)^a (2v)^b \cdot \indicator{u < v} \\
\label{eq:term2}\tag{**}+\quad&(2u)^a(2v)^b \cdot \Big(\indicator{x < y} - \indicator{u < v}\Big) \\
\label{eq:term3}\tag{***}+\quad&\left(x^ay^b - (2 u)^a  (2v)^b\right) \cdot \indicator{x < y}
\end{align}

Simplifying those terms separately, we have
\begin{align*}
\eqref{eq:term1} =\ &2^{a+b} \cdot \text{MDom}_{a,b}(u,v),\\
\eqref{eq:term2} =\ &x^a (y-1)^b \cdot \indicator{\feven(x) = \fodd(y)-1} = \textrm{MEq}_{a+b}(\feven(x), \fodd(y)-1),\\
\eqref{eq:term3}=\ &P_{a,b}(x,y)  \cdot \indicator{\fodd(x) < \feven(y)} + Q_{a,b}(x,y) \cdot \indicator{\feven(x) < \fodd(y)} + R_{a,b}(x,y) \cdot \indicator{\fodd(x) < \fodd(y)},
\end{align*}
where
$P_{a,b}(x,y) = (x^ay^b - (x - 1)^ay^b)$, $Q_{a,b}(x,y) = (x^ay^b - x^a(y-1)^b)$ and $R_{a,b}(x,y) = (x^ay^b - (x - 1)^a(y-1)^b).$

All in all, our recursion decomposes $\text{MDom}_{a,b}(x,y)$ into several terms -- either with the inputs reduced by a factor of $2$, the test for dominance replaced with a test for equality, or to monomials of smaller degree (observe that each of $P_{a,b}(x,y)$, $Q_{a,b}(x,y)$ and $R_{a,b}(x,y)$ is of degree at most $a+b-1$). Let $T(a,b,m)$ denote the number of instances of Hamming distance that a single instance of $\text{MDom}_{a,b}$, with inputs bounded in value by $2^m$, is reduced to. Since by  Lemma~\ref{lem:weighted_hamming}, $\text{MEq}_{a+b}$  reduces to $\bigo(m\cdot(a+b))$ instances of Hamming distance, there is
$$T(a,b,m) \le \bigo(m \cdot (a+b)) + T(a,b,m-1) +  \sum_{\substack{0 \le i \le a\\0 \le j \le b\\ (i,j) \not= (a,b)}} 3 T(i,j,m), $$
which is satisfied (for some constant $C$) by
$T(a,b,m) \le C \cdot m \cdot (a+b) \cdot {a+b+m \choose a,b,m} \cdot 4^a \cdot 4^b$.
For fixed values $a,b$ this is $\bigo( \log^{a+b+1} M )$. 
\end{proof}

\begin{proof}[Proof of Theorem~\ref{th:completeness1}]
Consider an arbitrary piecewise-polynomial binary function $\diamond$.
Consider its summand $P(x,y) \cdot \indicator{Ax+By+C>0}$.
If $A = 0$ or $B = 0$ then it reduces to a binary function of degenerated form $P(x,y) \cdot \indicator{Ax + C > 0}$ which in turn reduces to $\bigo(d^2)$
 multiplications by Lemma \ref{th:polynomials_are_easy}.

Otherwise, if $A \not= 0$ and $B \not= 0$, then there is a one-to-one linear input reduction, $u = -Ax$ and $v = By + C$, that reduces from $(-Ax)^i(By+C)^j \cdot \indicator{Ax + By + C> 0}$ to $u^i v^j \cdot \indicator{u<v}$. Note that any polynomial of degree $a$ and $b$ over $x$ and $y$ is a linear combination of $(-Ax)^i(By+C)^j$ for $0 \le i \le a$ and $0 \le j \le b$.

By applying those reductions to each summand and applying Theorem \ref{th:monomials} to each monomial of the summand, we reach the claimed bound.
\end{proof}

To prove Theorem~\ref{th:completeness2}, we need following technical Lemma:

\begin{lemma}
\label{lem:lines_lemma}
Consider a family of distinct lines $\Lambda = \{\lambda_i\}_{i=1}^{|\Lambda|}$, $\lambda_i = \{x,y : A_i x + B_i y + C_i = 0\}$ for integers $A_i,B_i,C_i$ such that $|A_i|,|B_i|,|C_i| \le M$. If there is at least one $\lambda \in \Lambda$ that is not axis-orthogonal, then
there exists $\lambda_i \in \Lambda$ and $\alpha,\beta,\gamma,\delta$ such that:
\begin{itemize}
\item for any line $\lambda_j$ that is not parallel to $\lambda_i$, the set $\{ ( \alpha x + \gamma,\beta  y + \delta) : x,y \in [0 \dots N]\}$ lies \emph{on the same side} of $\lambda_j$,
\item for any line $\lambda_j$ that is parallel to $\lambda_i$, the sets $\{ ( \alpha x + \gamma,\beta  y + \delta) : x > y \}$ and $\{ ( \alpha x + \gamma,\beta  y + \delta): x < y \}$ are separated by $\lambda_j$.
\end{itemize}
Moreover, $|\alpha|,|\beta|,|\gamma|,|\delta| \le \poly(M,N)$.
\end{lemma}

\begin{proof}
Pick $\lambda_i$ that is not axis-orthogonal, that is $A_i,B_i \not=0$. 

Let us denote the grid $\mathcal{G} = \{(\alpha x + \gamma, \beta y + \delta) : x,y \in [0 \dots N]\}$. To guarantee that main diagonal of $\mathcal{G}$ lies on $\lambda_i$, we need to have $\alpha = B_i \cdot k$ and $\beta = -A_i \cdot k$ for some nonzero integer $k$, and select values of $\gamma,\delta$ accordingly so that $(\gamma,\delta) \in \lambda_i$.

For non-parallel $\lambda_i,\lambda_j$, the coordinates of intersection point are:
$$x_{i,j} = - \begin{vmatrix}C_i & B_i \\ C_j & B_j\end{vmatrix} / \begin{vmatrix}A_i & B_i \\ A_j & B_j\end{vmatrix} \quad \quad y_{i,j} = - \begin{vmatrix}A_i & C_i \\ A_j & C_j\end{vmatrix} / \begin{vmatrix}A_i & B_i \\ A_j & B_j\end{vmatrix}.$$

To guarantee that whole $\mathcal{G}$ lies on the same side of $\lambda_j$, it is enough to make sure that all $4$ corners are on the same side. However, we observe that iff e.g. corners $(\gamma, \delta)$ and $(\alpha N + \gamma, \delta)$ are separated by $\lambda_j$, it means that for some $r \in [0,1]$ lines  $\lambda_j$ and $A_i (x-r\alpha N) + B_i y + C_i = 0$ (that is $\lambda_i$ shifted in $x$ by $+r \alpha N$) intersect on point with $x = \delta$. 
To satisfy the first condition of the lemma, it is enough if every point of the convex closure of $\mathcal{G}$ has $x$ coordinate with absolute value at least $2M^2 + |\alpha| M N$, since that is larger than any possible intersection point as described above (condition (a)). Similarly for $y$ coordinate it should be at least $2M^2 + |\beta| MN$ (condition (b)).

Take $\lambda_j$ parallel to $\lambda_i$, that is they differ only on value of $C$. We first make sure  that all such $\lambda_j$ fall between lines $\{(\alpha t + \gamma, \beta (t+1) + \delta) : t \in \mathbb{R}\}$ and $\lambda_i$ or $\lambda_i$ and $\{(\alpha (t+1) + \gamma, \beta t + \delta) : t \in \mathbb{R}\}$ (those lines are $\lambda_i$ ``shifted'' one step up or down in the grid), by making sure $\alpha$ and $\beta$ are large enough in absolute value. Indeed, it is enough to have $|\alpha A_i| = |\beta B_i| > 2 M$ being largest possible difference between two values of $C$. It is enough to select $k = 3M$, and $\alpha = 3MB_i, \beta = -3MA_i$.

We then select $\gamma$ and $\delta$ as smallest in absolute value points of $\lambda_i$ such that conditions (a) and (b) are satisfied. 
\end{proof}

\begin{proof}[Proof of Theorem~\ref{th:completeness2}]
Let us take the binary function $x\diamond y = \sum_i P_i(x,y) \cdot \indicator{A_i x + B_i y + C_i > 0}$ as in the theorem statement, assuming it is of the simplest form (no redundant terms and minimal number of summands possible).  We construct a reduction from Hamming distance to $\diamond$ by a series of intermediate operators. 

Let $d$ be the highest degree of any $P_1,P_2,\ldots$. Consider all the lines being borders of regions, that is $\lambda_i = \{ (u,v) : A_i u + B_i v + C = 0 \}$ (as elements of the continuous Euclidean plane). 

We now apply Lemma~\ref{lem:lines_lemma}, with $N = 3dM+2d$. Consider $F(x,y) \defeq (\alpha x + \gamma) \diamond (\beta y + \delta)$. Limited to $x,y \in [0 \dots N]$, $F(x,y)$ is piecewise linear of a much simpler form:
$$F(x,y) = Q_{>}(x,y) \cdot \indicator{x > y} + Q_{=}(x,y)\cdot \indicator{x = y} + Q_{<}(x,y)\cdot\indicator{x < y}$$
for $Q_{>},Q_{=},Q_{<}$ being polynomials of degree at most $d$, and $Q_{<}\not\equiv Q_{>}$. Let $D_x, D_y$ be the operators of discrete differentiation, that is $D_x F(x,y) \defeq F(x+1,y) - F(x,y)$, $D_y F(x,y) \defeq F(x,y+1) - F(x,y)$. There are integers $0 \le a,b \le d$ such that $D_x^a D_y^b (Q_{<}(x,y) - Q_{>}(x,y)) \equiv c$ for some constant $c \not= 0$. Thus if we consider the function:
$$G(x,y) \defeq \frac1c \cdot D_x^a D_y^b (F(x,y) - Q_{>}(x,y)),$$
it has the following properties on $x,y \in [0 \dots N-d]$: for $y-x > d$: $G(x,y) = 1$, and for $y - x < -d$: $G(x,y) = 0$. We observe that for $x,y \in [0 \dots M]$, there is
$\dom(x,y) = G(3d\cdot x,3d\cdot y+d).$
All in all, $\ham$ reduces to $\bigo(d^2)$ instances of $\diamond$ and a single evaluation of a fixed polynomial $Q_{>}(x,y)$, which reduces to $\bigo(d^2)$ multiplications.
\end{proof}

\section{Sparse matrix multiplication and \textsc{AllPairs}- problems}
\label{sec:smm}
We devote this section to proving Theorem~\ref{th:sparsity_hamming}. We start with the following reduction, which we believe to be a folklore result. Here, by 0/1 matrices we mean matrices with integer entries being either 0 or 1 (but all arithmetic is performed in the ring $\mathbb{Z}$). Let $|\cdot|$ denote the number of nonzero entries in a matrix/vector/set of entries, and $A_{i*}$ and $A_{*i}$ denote $i$-th row and $i$-th column of $A$, respectively.

\begin{lemma}[folklore]
\label{lem:01matrices}
Multiplication of (sparse) integer matrices has the same complexity as multiplication of (sparse) 0/1 matrices (up to $\poly\log n$ factors).
\end{lemma}

\begin{proof}
Consider the multiplication of two integer matrices with nonnegative entries $A \times B$, bounded in value by $M$. For integer $k$ we define $\textrm{bit}_k(x)$ to be the value of $k$-th bit of $x$. Denote $A_k = \textrm{bit}_k(A)$ to be the 0/1 matrix selecting $k$-th bit of $A$ entries, and $B_k = \textrm{bit}_k(B)$. Consider 

\begin{align*}
A' &= \begin{bmatrix} A_{0} \\ \vdots \\ A_{\log M} \end{bmatrix} \in \mathbb{Z}^{ (n \log M) \times n} \text{ and } B' = \begin{bmatrix} B_{0}&\cdots&B_{\log M} \end{bmatrix} \in \mathbb{Z}^{ n \times (n \log M)} \text{ and their product }\\
A' \times B' &= \begin{bmatrix} A_0 \times B_0& \cdots& A_0 \times B_{\log M}\\ \vdots& \ddots& \vdots\\ A_{\log M} \times B_0& \cdots& A_{\log M} \times B_{\log M} \end{bmatrix} \in \mathbb{Z}^{ (n \log M) \times (n \log M) }.
\end{align*}

Since $A = \sum_i 2^i A_i$ and $B = \sum_i 2^i B_i$, there is $A \times B = \sum_i \sum_j 2^{i+j} A_i \times B_j$, meaning that $A \times B$ follows from the product of two 0/1 matrices of dimensions that are larger by a factor of $\bigo(\log n)$. To get rid of the nonnegativity assumption, we can represent any integer matrices $A,B$ as $A=A_1 - A_2$ and $B=B_1-B_2$ where $A_1,A_2,B_1,B_2$ are nonnegative, and consider the product $\begin{bmatrix} A_1\\ A_2 \end{bmatrix} \times \begin{bmatrix} B_1 \ B_2 \end{bmatrix}$.
\end{proof}

An easy reduction shows that \AllPairsHammingShort reduces to the multiplication of sparse 0/1 rectangular matrices. This follows in spirit the ideas used in \cite{DBLP:conf/soda/Yuster09} (where it was used in the context of dominance), but instead of packing only the ``dense'' part of the computation into a matrix multiplication problem, we put it all.

\begin{lemma}[c.f. \cite{DBLP:journals/talg/YusterZ05}~Lemma 3.2.]
\label{yusterzwicklemma}
Consider multiplication $\sparse(n_1, N, n_2; m_1, m_2)$ of sparse matrices $A$ and $B$, and permutation $\pi$ is so that $|A_{*\pi(i)}| \cdot |B_{\pi(i)*}|$ are non-increasing with $i$. Then for any $1 \le \ell \le N$ there is $\sum_{j>\ell} |A_{*\pi(j)}| \cdot |B_{\pi(j)*}| \le \frac{m_1,m_2}{\ell}$.
\end{lemma}

\begin{theorem}
\label{th:hamming_to_sparse}
\AllPairsHammingShort on vectors of dimension $d$ reduces deterministically to $\sparse(n,N,n;M;M)$ for some $M,N = \bigo(n d)$. Furthermore, for $d<n$, $N$ can be as small as $\bigo(d^2)$.
\end{theorem}

\begin{proof}
Let $U = \begin{bmatrix} u_1& \cdots&u_n \end{bmatrix}^T$ and $V=  \begin{bmatrix}v_1&\cdots&v_n \end{bmatrix}^T$ for $u_1,\ldots,u_n,v_1,\ldots,v_n \in \mathbb{Z}^d$. W.l.o.g. we assume that entries of those vectors are actually from $[2n]$ -- if it is otherwise, we can scan each coordinate separately and rename the entries.  Consider $A \in \{0,1\}^{n \times N}$ for $N = 2nd$, defined as such: $A[i, j + k \cdot d] = 1$ iff $U[i,j] = k+1$, for $1 \le i \le n$, $1 \le j \le d$ and $0 \le k < 2n$. Similarly, we construct $B \in \{0,1\}^{n \times N}$ from $V$.

We now observe that $C = A \times B^T$ allows us to compute $\mprod(\ham,U,V)$, since for any $i,j \in [n]$, $C[i,j] = d - \ham(u_i,v_j)$.

Now to reduce the value of $N$, we Lemma~\ref{yusterzwicklemma}: rearrange columns of $A$ and the rows of $B$ simultaneously so that $|A_{*i}| \cdot |B_{i*}|$ is non-increasing (this preserves the product). We then truncate $A$  to keep only the first $d^2$ columns and truncate $B$ to keep only the first $d^2$ rows. By Lemma~\ref{yusterzwicklemma}, the time needed to compute the contribution of the truncated rows/columns is $\bigo(|A| \cdot |B| / d^2) = \bigo(n^2)$. 
\end{proof}

Observe that regardless of whether one solves the \AllPairsHammingShort instance by adapting \cite{DBLP:conf/soda/Yuster09}, or by using Theorem~\ref{th:hamming_to_sparse} and \cite{DBLP:journals/talg/YusterZ05}, the resulting computation is roughly similar, thus it is no surprise that the resulting runtime is identical. We now present a converse argument, that the multiplication of \emph{arbitrary} sparse matrices is no harder than the corresponding \AllPairsHammingShort. Here the reduction is a little bit more tricky, since e.g.~the 0/1 matrices resulting from Theorem~\ref{th:hamming_to_sparse} have a combinatorial inner structure that arbitrary instances of matrix multiplication might not have.

\begin{theorem}
\label{th:sparse_to_hamming}
For $N \ge n$, the multiplication problem in $\sparse(n,N,n;N,N)$ reduces under a randomized (Las Vegas) reduction to an \AllPairsHammingShort instance with $\widetilde\bigo(n)$ vectors of dimension $\widetilde\bigo(N/n)$.
\end{theorem}
\begin{proof}
Let $A,B$ be the input 0/1 matrices. W.l.o.g. $N$ is divisible by $n$, as if it is not the case, we round $N$ up to the nearest multiplicity of $n$ and pad $A$ and $B$ with zeroes accordingly. Denote $d = N/n$. As a first step, we pick uniformly $\pi \in [d]^{[N]}$, which we use to decide for columns of $A$ (rows of $B$) contribute to which columns of output $U$ (rows of $V$) they contribute.

We want to construct matrix $U$, such that if $A[i,j] = 1$ then $U[i,\pi(j)] = j$. However, such mapping might not be well defined, as there might be conflicts of the form $j_1,j_2$ such that $A[i,j_1] = A[i,j_2] = 1$ and $\pi(j_1) = \pi(j_2)$. We deal with conflicts by row-splitting.

Let $r_{i,k} = \sum_{j : \pi(j) = k} A[i,j]$  be the number of ones that are mapped to a given $i,k$ cell. Denote $c_i = \max_j r_{i,j}$. Then $i$-th row of $A$ is split into $c_i$ rows in $U$, i.e. rows $C_i+1, \ldots, C_i+c_i$ where $C_i = c_1 + \ldots + c_{i-1}$. Then, for any $i,j$ if $A[i,j]=1$, then $U[C_i + t, \pi(j)] = j$, where $A[i,j]$ was $t$-th value 1 cell among all $A[i,x]$ such that $\pi(x) = \pi(j)$. 

To complete the construction of $U$, any value not yet set after processing all of $A$ is assigned an unique value in its column.

Observe that since $\pi$ was picked at random we have that $\mathbb{E}[r_{i,j}] = |A_{i*}|/d$. By Chernoff bound, w.h.p. $|r_{i,j}| = \bigo(\log n/\log \log n) \cdot \lceil |A_{i*}| / d\rceil$. Denote by $c_i = \max_{j} |r_{i,j}|$. As we split the $i$-th row of $A$ into $c_i$ rows in $U$, we bound the total number of rows in $U$ as
$\sum_i c_{i} = \bigo(\log n/\log \log n) \cdot \sum_{i} \lceil |A_{i*}| / d\rceil = \bigo(\log n/\log \log n) \cdot (n + |A|/d) = \bigo(n \log n/\log \log n)$.

The construction of $V$ from $B$ follows, switching row and column roles in the presented reduction.

 Let the $i$-th row of $A$ is mapped to rows $C_i+1, \ldots, C_i+c_i$ in $U$, and the $j$-th column of $B$ is mapped to columns $D_j+1, \ldots, D_j + d_j$ in $V$. It follows that $\mprod(\ham,U,V)$ encodes $A \times B$ since
 $$(A \times B)[i,j] = \sum_{a=1}^{c_i} \sum_{b=1}^{d_j} \left(d-\mprod(\ham,U,V)[C_i + a, D_j + b]\right).$$

To finish the argument, we observe that if $A,B$ are provided in a compressed form (which they need to, as an explicit representation is already too large), the $U$ and $V$ can be generated without any significant additional computational overhead, in time $\widetilde\bigo(|A|+|B|+N)$.
\end{proof}

We can use the same techniques to derive a relation between \AllPairsHammingShort on \emph{sparse inputs} with sparse matrix multiplication. We obtain the following:

\begin{corollary}
\label{cor:apham_sparse}
\AllPairsHammingShort on inputs $A$,$B$ of size $n$, with $m_1$ and $m_2$ relevant entries, respectively, takes $\bigo(\sparse(n, \frac{m_1m_2}{n^2}, n; m_1, m_2))$ time. 
\end{corollary}

\begin{proof}
Following the reasoning from the proof of Theorem~\ref{th:hamming_to_sparse}, we construct an instance of sparse matrix multiplication with parameters $\sparse(n,N,n;m_1,m_2)$ for some large integer $N$. However, to go from counting ``matches'' which this 0/1 matrix multiplication does, to counting mismatches, we need to count the number of aligned relevant entries between $A$ and $B$. This is done with a single multiplication of sparse matrices in time $\sparse(n, d, n; m_1, m_2)$. By Lemma~\ref{yusterzwicklemma} dimension $d$ is reduced to $m_1 m_2/n^2$ at the cost of $\bigo(n^2)$ additional computation.
\end{proof}

We also obtain explanation why for problems like \AllPairsHammingShort we observe similarly looking trade-offs for $d$ vs. time and for sparsity vs. time (since with proper fixing of parameters, they both reduce to the same type of sparse matrix multiplication instances).
\begin{corollary}
For $d\le n$, complexity of \AllPairsHammingShort on $n$ vectors of dimension $d$ is within polylog factor from \AllPairsHammingShort on $n$ vectors with $n\cdot d$ relevant entries.
\end{corollary}

We linked the complexity of \AllPairsHammingShort and sparse matrix multiplication. However, to improve the current upperbounds, one needs to improve the sparse matrix multiplication upperbound for almost square matrices -- this follows from allocating $\bigo(n^{\rho-\varepsilon})$ runtime for row/column elimination procedure.
\begin{corollary}
Any improvement to the exponent of $\sparse(n, n^{4 - \rho + \varepsilon}, n; n^2, n^2)$ beyond \cite{DBLP:journals/talg/YusterZ05} method runtime would improve the exponent of \AllPairsHammingShort.
\end{corollary}
Current bounds imply that one needs to improve $\sparse(n, n^{1.3167 + \varepsilon}, n; n^2, n^2)$.

\section{Conclusion}
\label{sec:concl}
There are several immediate applications of Theorem~\ref{th:completeness1} and Theorem~\ref{th:completeness2}. 
The first one is that the improvement to \AllPairsDominanceShort from \cite{DBLP:conf/soda/Yuster09} translates to other \textsc{AllPairs} problems:
\begin{corollary}  \AllPairsDominanceShort, \AllPairsLOneDistancesShort, \AllPairsLoddDistancesShort, \AllPairsThresholdShort, \AllPairsHammingShort and $(+,\min)$-\textsc{MatrixProduct} are solvable in time $\widetilde\bigo(n^\rho)$, where $\rho \le 2.6834$ is a solution to $\rho = \omega(1, 4-\rho, 1)$.
\end{corollary}

Observe that the reductions we presented map $\star$ to $\star$. Thus, e.g.~by \cite{vassilevska2008efficient},\cite{DBLP:journals/toc/VassilevskaWY09} and \cite{DBLP:conf/soda/DuanP09}, we immediately get that all considered \textsc{AllPairs} problems are of the same complexity even on sparse inputs, up to a $\poly\log n$ multiplicative term and a $\sparse(n,n,n;m_1,m_2)$ additive term.
\begin{corollary}
Consider sparse inputs where we denote by $m_1$ and $m_2$ the number of entries in $A$ and $B$ that contribute to the score, where $A$ and $B$ are matrices of $n$ vectors of dimension $n$. \AllPairsDominanceShort, \AllPairsLOneDistancesShort, \AllPairsLoddDistancesShort, \AllPairsThresholdShort, \AllPairsHammingShort and $(+,\min)$-\textsc{MatrixProduct} are solvable in time
$\widetilde\bigo(\min(n^\omega +\sqrt{m_1 m_2} \cdot n^{\frac{\omega - 1}{2}},n^2+(m_1 m_2)^{\frac{\omega - 2}{\omega - \alpha - 1}} n^{\frac{2 - \alpha \omega}{\omega - \alpha - 1}})).$
\end{corollary}

Since our reductions preserve the dimension of the problems, any tradeoff between $d \ll n$ and the runtime translates to all other problems as well, with a $\poly\log n$ multiplicative term and a $\widetilde\bigo(n^\omega)$ additive term. One can improve the runtime of the algorithm presented in \cite{DBLP:conf/cocoon/MinKZ09} using the trick of batch-processing via rectangular matrix multiplication in \cite{DBLP:conf/soda/Yuster09}, as done for Dominance Product in \cite{DBLP:conf/isaac/GoldS17}, to obtain the following time complexity:
\begin{corollary}
For $n$ vectors of dimension $d = n^s$ for $0 \le s \le 1$, \AllPairsDominanceShort, \AllPairsLOneDistancesShort, \AllPairsLoddDistancesShort, \AllPairsThresholdShort, \AllPairsHammingShort and $(+,\min)$-\textsc{MatrixProduct} are solvable in time
$\widetilde\bigo(n^{\rho(s)})$
where $\rho(s) = \inf\{ x : 2 \le x \le 3 \text{ and } \omega(1,2+2s-x,1) \ge x\}$.
 In particular, for $d = \bigo(n^{\alpha/2}) \supseteq \bigo(n^{0.156945})$ all those problems are solvable in time $\widetilde\bigo(n^2)$.
\end{corollary}

Similarly, one can look into the relation between sparsity and runtime for pattern matching problems. Here, we obtain the following result
\begin{theorem}
\label{th:sparse_pm}
For a text of length $n$ and a pattern of length $m$, $n\ge m$, with $s_t$ and $s_p$ relevant entries, respectively, the runtime of \PMHammingShort, \PMLessThanShort, \PMThresholdShort and \PMLoddShort is $\widetilde\bigo(\sqrt{n  s_t s_p} + n)$.
\end{theorem}

\begin{proof}
	Consider \PMLessThanShort. The proof follows the non-sparse case. W.l.o.g. all the $2 s_p$ actual entries are distinct integers (if it is not so, they can be made so using small $\varepsilon>0$ shifts and then re-arranged back into integers preserving order). The $s_p$ relevant entries of the pattern are sorted and partitioned into $k$ buckets $B_1,\ldots,B_k$ so that $B_1$ gets $s_p/k$ smallest elements, $B_2$ following $s_p/k$ smallest elements, etc. We get inter-bucket contribution for bucket $B_i$ from convolution of $P_i^R$ with $T_i$, where $P_i,T_i$ are binary strings such that $P_i[j] = 1$ iff $P[j] \in B_1 \cup \ldots \cup B_i$ and $T_i[j]=1$ iff $T[j] \in B_i$. This in a total takes $\bigo(kn\log m)$ time for all $k$ buckets. Intra-bucket contributions are captured in a brute force manner in $\bigo(s_t s_p/k)$ where each relevant text element is compared with at most $s_p/k$ elements in its corresponding bucket. Choosing $k$ to be $\max(1, \sqrt{(s_t s_p)/(n\log m)})$ gives the time bound of $\widetilde\bigo(n+\sqrt{n s_t s_p})$.
\end{proof}

We present the following application of the scaling/filtering framework: weighted mismatches. We distinguish between \emph{position weighted mismatches} and \emph{character weighted mismatches}.
In the pattern matching setting, the former asks for $\mathbf{O}[i] = \sum_{j: \mathbf{P}[j] \not= \mathbf{T}[i+j]} w(j)$, whereas the latter asks for $\mathbf{O}[i] = \sum_{j: \mathbf{P}[j] \not= \mathbf{T}[i+j]} w(\mathbf{P}[j])$, for some given weight function $w : \mathbb{Z} \to \mathbb{Z}$ .  We see that character weighted mismatches are expressible by a function $w(x)\cdot \indicator{x \not= y}$ and get by Lemma~\ref{lem:weighted_hamming} that Hamming Distance Pattern Matching with Character Weights is no harder than \PMHammingShort (up to a $\log n$ factor). For position weights, we present the following:
\begin{theorem}
\label{th:weighted_ham}
Hamming Distance Pattern Matching with Position Weights reduces to $\bigo(\log n)$ instances of \PMHammingShort.
\end{theorem}

\begin{proof}
We solve $\bigo(\log n)$ instances of \PMHammingShort with filtering involved. This is done by constructing different pattern strings where $\mathbf{P}_i$ is defined as follows:
$$\mathbf{P}_i[j] = \begin{cases}P[j]&\quad i\text{-th bit of }w(j)\text{ is }1\\\star&\quad\text{otherwise.}\end{cases} $$ 
Let $\mathbf{O}_i$ be the result vector of \PMHammingShort between text $\mathbf{T}$ and pattern $\mathbf{P}_i$. The final result vector, $\mathbf{O}$, for the Hamming distance pattern matching with position weights can be computed such that $\mathbf{O}[k] = \sum_{i = 0}^{\ceil{\log W}} 2^i \cdot \mathbf{O}_i[k]$ where $W$ is the maximum position weight. Given our assumption that $W = \poly(n)$, the result follows.
\end{proof}

What remains is to show that one can get rid of $\star$ when e.g. computing Hamming distance. We show this in the pattern matching setting for simplicity. However this can be easily extended for matrix multiplication problems as well.

\begin{lemma}
\label{lem:nonsparse_hamming}
	Hamming distance in $\mathbb{N}+\{\star\}$ reduces preserving linearity to two instances of Hamming distance in $\mathbb{N}$.
\end{lemma}
\begin{proof}
	Let $x,y \in \mathbb{N}+\{\star\}$. To compute $\ham(x,y)$, we first use mapping that puts $\star$ into separate integer, and then apply correction that fixes distances between $\star$.
	
	For the first instance:
	$$f(t) = \begin{cases}0 \quad&\text{if } t = \star \\ t+1\quad&\text{otherwise } \end{cases}$$
	
	As for the second instance:
	$$g(t) = \begin{cases}0 \quad&\text{if } t = \star \\ 1\quad&\text{otherwise } \end{cases}$$
	Observe that $\ham(x,y) = \ham(f(x),f(y)) - \ham(g(x),g(y)).$
	
\end{proof}

While it is no surprise that for example the technique of~\cite{DBLP:conf/soda/Yuster09} can be applied to other \textsc{AllPairs} problems, it is a nice side effect of our reduction that it can be applied ``automatically'' without looking deeper into the structure of any of the different \textsc{AllPairs} problems involved.
The reductions presented signify that regardless of whether we are looking for improved upper bounds, or new lower bounds, it is enough to concentrate on a single score function from the whole class of equivalent functions. In our opinion, Hamming distance is the ``cleanest'' score function, since it is the simplest -- it assumes no arithmetic underlying structure of the alphabet (unlike e.g.~$L_1$ distance) and not even an ordering of the alphabet. 

\bibliographystyle{alpha}
\bibliography{bib}

\appendix

\section{Supplementary reductions}
\begin{theorem}
\label{lem:l1_equiv_min}
$L_1$ distance reduces to $\min$ and multiplications. $\min$ reduces to $L_1$ and multiplications.
\end{theorem}
\begin{proof}
$\min(x,y) = x/2+y/2 - |x-y|/2$ and $|x-y| = x+y-\min(x,y)$.
\end{proof}
\begin{lemma}\label{lem:Dom_eq_Thr}
Dominance and $\delta$-threshold are equivalent.
\end{lemma}
\begin{proof}
Since both dominance and threshold are shift-invariant, we assume $0 \le x,y \le M$ for some $M$ bounded by $\poly(n)$.
Dominance reduces to one instance of threshold as $\dom(x,y) = \thr_{\delta}(x+\delta,y)$ for any $\delta>M$.
Threshold reduces to two instances of dominance as $\thr_\delta(x,y) = \dom(y+\delta, x) + \dom(x+\delta,y)$ for $\delta>0$.
\end{proof}

\begin{lemma}[\cite{vassilevska2008efficient}]
Hamming distance reduces to 2 instances of dominance.
\end{lemma}
\begin{proof}
$\ham(x,y) = \dom(x+1,y) + \dom(-x+1,-y)$.
\end{proof}

\begin{lemma}[\cite{DBLP:journals/ipl/LipskyP08a}]
Dominance reduces to 2 instances of $L_1$, Hamming distance reduces to 3 instances of $L_1$.
\end{lemma}
\begin{proof}
$\dom(x,y) = |x-(y+1)|/2-|x-y|/2+1/2$ and
$\ham(x, y) = 1+|x-y|- |x-(y+1)|/2 - |(x+1)-y|/2$
\end{proof}

\section{Example reductions}
\label{sec:examples}
Figures~\ref{fig:example-l1-to-less-than} and~\ref{fig:example-less-than-to-ham} illustrate our reductions from Theorems~\ref{th:L1_to_Dominance} and~\ref{th:Dominance_to_Hamming}, respectively.

\begin{figure}[!h]
	\centering
	\includegraphics[width=0.95\textwidth]{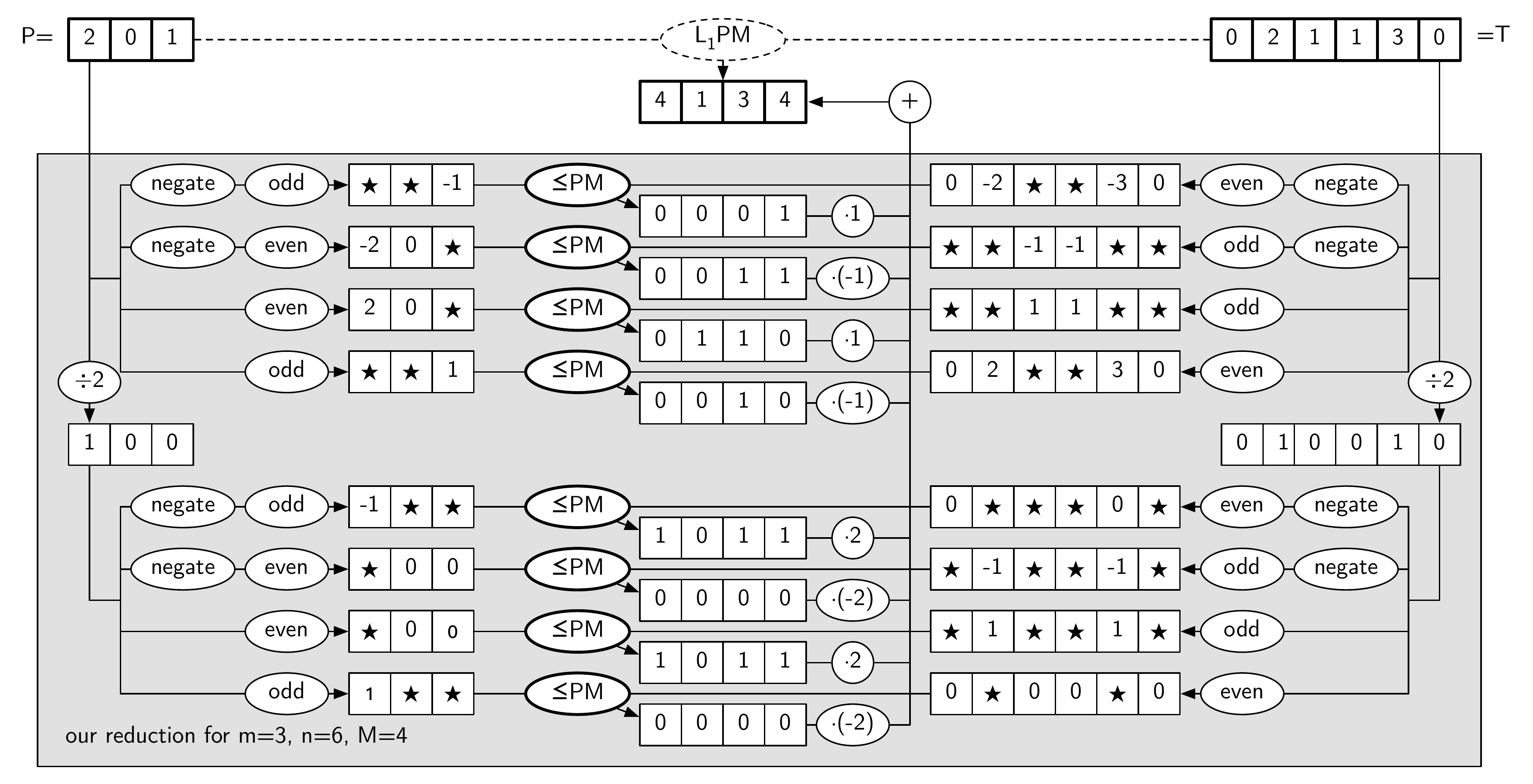}
	\caption{Our reduction from \PMLOne to \PMLessThan in Theorem~\ref{th:L1_to_Dominance} instantiated for a pattern $P$ of length $m=3$ and a text $T$ of length $n=6$ over an alphabet of integers $\{0,1,2,3\}$, so $M=2^2=4$.}
	\label{fig:example-l1-to-less-than}
\end{figure}

\begin{figure}[!h]
	\centering
	\includegraphics[width=0.95\textwidth]{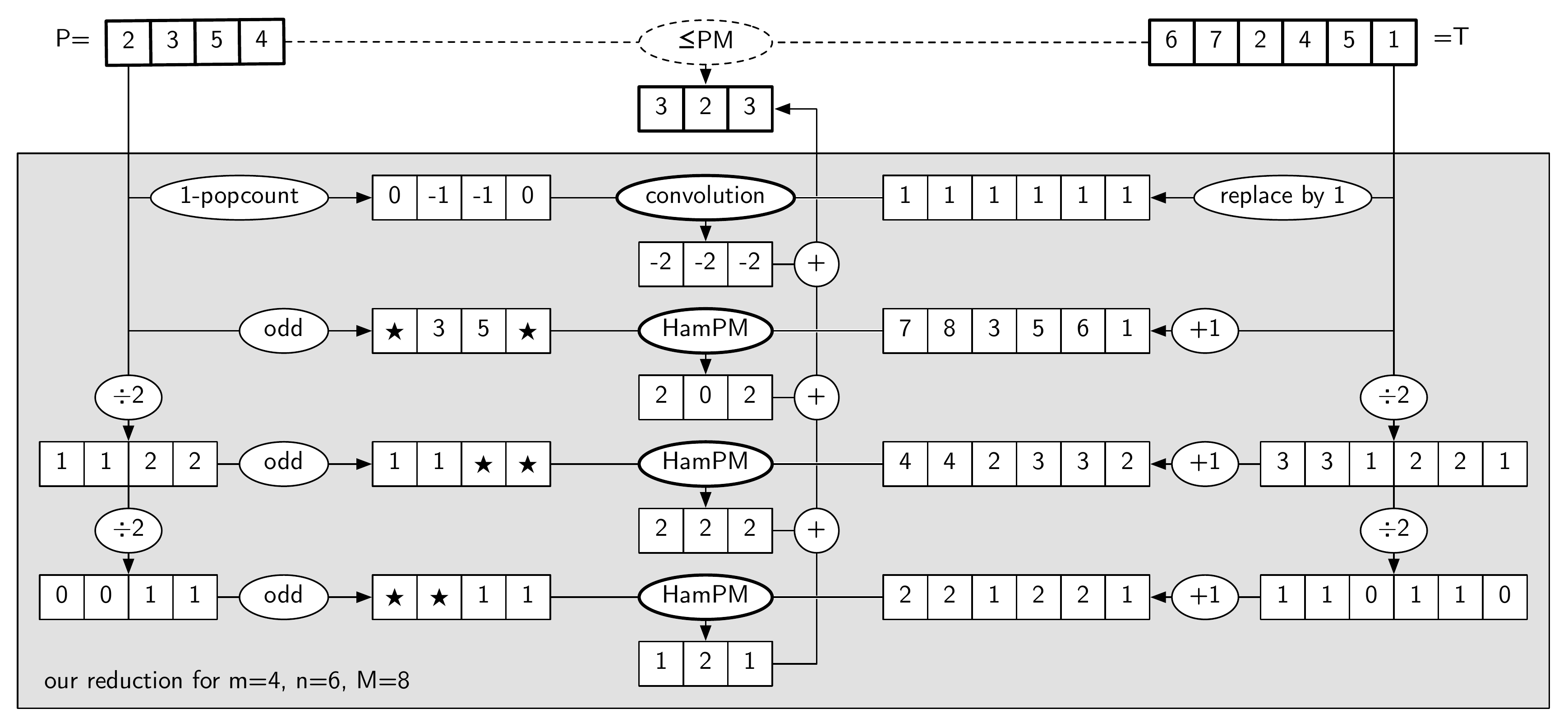}
	\caption{Our reduction from \PMLessThan to \PMHamming in Theorem~\ref{th:Dominance_to_Hamming} instantiated for a pattern $P$ of length $m=4$ and a text $T$ of length $n=6$ over an alphabet of integers $\{0,1,2,3,4,5,6,7\}$, so $M=2^3=8$.}
	\label{fig:example-less-than-to-ham}
\end{figure}

\end{document}